\documentclass[11pt]{article}

\usepackage{graphicx}
\usepackage{amssymb}
\usepackage{amsmath}
\usepackage{amsthm}
\usepackage{color}

\begin{document}
\newcommand{\hide}[1]{}

\newcommand{\largespace}{\hspace*{1.0in}}

\newcommand{\chs}[2]{\left(\!
    \begin{array}{c}
      #1 \\
      #2
    \end{array}
  \!\right)}

\newcommand{\bba}{\mathbf{a}}
\newcommand{\bbb}{\mathbf{b}}
\newcommand{\bbc}{\mathbf{c}}
\newcommand{\bbd}{\mathbf{d}}
\newcommand{\bbe}{\mathbf{e}}
\newcommand{\bbf}{\mathbf{f}}
\newcommand{\bbg}{\mathbf{g}}
\newcommand{\bbh}{\mathbf{h}}
\newcommand{\bbi}{\mathbf{i}}
\newcommand{\bbj}{\mathbf{j}}
\newcommand{\bbk}{\mathbf{k}}
\newcommand{\bbl}{\mathbf{l}}
\newcommand{\bbm}{\mathbf{m}}
\newcommand{\bbn}{\mathbf{n}}
\newcommand{\bbo}{\mathbf{o}}
\newcommand{\bbp}{\mathbf{p}}
\newcommand{\bbq}{\mathbf{q}}
\newcommand{\bbr}{\mathbf{r}}
\newcommand{\bbs}{\mathbf{s}}
\newcommand{\bbt}{\mathbf{t}}
\newcommand{\bbu}{\mathbf{u}}
\newcommand{\bbv}{\mathbf{v}}
\newcommand{\bbw}{\mathbf{w}}
\newcommand{\bbx}{\mathbf{x}}
\newcommand{\bby}{\mathbf{y}}
\newcommand{\bbz}{\mathbf{z}}

\newcommand{\bbA}{\mathbf{A}}
\newcommand{\bbB}{\mathbf{B}}
\newcommand{\bbC}{\mathbf{C}}
\newcommand{\bbD}{\mathbf{D}}
\newcommand{\bbE}{\mathbf{E}}
\newcommand{\bbF}{\mathbf{F}}
\newcommand{\bbG}{\mathbf{G}}
\newcommand{\bbH}{\mathbf{H}}
\newcommand{\bbI}{\mathbf{I}}
\newcommand{\bbJ}{\mathbf{J}}
\newcommand{\bbK}{\mathbf{K}}
\newcommand{\bbL}{\mathbf{L}}
\newcommand{\bbM}{\mathbf{M}}
\newcommand{\bbN}{\mathbf{N}}
\newcommand{\bbO}{\mathbf{O}}
\newcommand{\bbP}{\mathbf{P}}
\newcommand{\bbQ}{\mathbf{Q}}
\newcommand{\bbR}{\mathbf{R}}
\newcommand{\bbS}{\mathbf{S}}
\newcommand{\bbT}{\mathbf{T}}
\newcommand{\bbU}{\mathbf{U}}
\newcommand{\bbV}{\mathbf{V}}
\newcommand{\bbW}{\mathbf{W}}
\newcommand{\bbX}{\mathbf{X}}
\newcommand{\bbY}{\mathbf{Y}}
\newcommand{\bbZ}{\mathbf{Z}}

\newtheorem{theorem}{Theorem}[section]
\newtheorem{corollary}{Corollary}[theorem]
\newtheorem{lemma}[theorem]{Lemma}
\newtheorem{definition}[theorem]{Definition}
\newtheorem{claim}[theorem]{Claim}

\newcommand\numberthis{\addtocounter{equation}{1}\tag{\theequation}}

\newcommand{\rjc}[1]{{\color{blue}{#1}}}
\newcommand{\RJC}[1]{\textbf{\color{Orange}RJC: #1}}
\newcommand{\yxt}[1]{{\color{red}{#1}}}
\newcommand{\YXT}[1]{\textbf{\color{green}YXT: #1}}


\title{On the Existence of Pareto Efficient and Envy-Free Allocations\footnote{This work was supported in part by NSF grants CCF-1527568 and CCF-1909538.}}


\author{Richard Cole and Yixin Tao}
\maketitle

\begin{abstract}
Envy-freeness and Pareto Efficiency are two major goals in welfare economics. The existence of an allocation that satisfies both conditions has been studied for a long time. Whether items are indivisible or divisible, it is impossible to achieve envy-freeness and Pareto Efficiency ex post even in the case of two people and two items. In contrast, in this work, we prove that, for any cardinal utility functions (including complementary utilities for example) and for any number of items and players, there always exists an ex ante mixed allocation which is envy-free and Pareto Efficient, assuming the allowable assignments are closed under swaps, i.e. if given a legal assignment, swapping any two players’ allocations produces another legal assignment. The problem remains open in the divisible case.

We also investigate the communication complexity for finding a Pareto Efficient and envy-free allocation.

\end{abstract}
\section{Introduction}
Efficiency and fairness are two important goals in welfare economics.
Pareto Efficiency and envy-freeness are the foremost notions of, respectively, efficiency and fairness for the allocation problem.
A given allocation is Pareto Efficient if there is no other allocation
in which no one loses and at least one person gains,
and it is envy-free if no person can gain by exchanging her
allocation with someone else's.

The question of whether there exists an allocation that is both Pareto Efficient and envy-free has been studied for a long time.
Unfortunately, for general utility functions, in both the divisible and indivisible cases, solutions that are simultaneously Pareto Efficient and envy-free cannot be ensured. In the indivisible case, in which items cannot be split, allocating one item among two people who both value the item will never be envy-free and, in the divisible case, there is a well-known example which comprises two items and two players such that there is no simultaneously Pareto Efficient and envy-free allocation.

However, in these counter-examples, the allocations are deterministic. In other words, there is no randomness. Randomness is often the enabler for existence; for instance, the existence of a Nash equilibrium. So, for the allocation problem, what happens if we consider  a mixed allocation instead of a pure allocation? Does there exist a mixed Pareto Efficient and envy-free allocation?

In this paper, we focus on the indivisible case, and our answer is YES. We prove that, for any cardinal utility functions and for any number of items and players, there always exists an ex ante mixed allocation which is envy-free and Pareto Efficient, assuming the allowable assignments are \emph{swappable}.
An allocation set is swappable if the allocation that results from any single pair
of players exchanging their allocated bundles is also allowable.
Clearly, the allocation set that can allocate any subset of items to any player is \emph{swappable}.

Our approach is to construct a mapping from the space of mixed allocations and weight vectors to itself.  We then apply the Kakutani fixed-point theorem \cite{kakutani1941generalization} to obtain a fixed point. Finally, we prove that the fixed point corresponds to a mixed Pareto Efficient and envy-free allocation. The proof is inspired by \cite{svensson1983existence, barbanel2005geometry, weller1985fair}.

We next ask how readily one can calculate a Pareto Efficient and envy-free allocation, in terms of the unavoidable communication cost, referred to as the communication complexity of the problem. In particular, we show that even with two players, if their utility functions are submodular, in general, calculating such an allocation will require $\Omega(2^\frac{m}{2})$ bits to be communicated between the players, where $m$ is the number of items to be allocated, which rapidly becomes infeasible as $m$ increases.

\section{Related Work}
A detailed survey on fairness and further background can be found in  \cite{brams1996fair, brandt2016handbook, moulin2004fair}.

Research on fair allocation research dates back to at least \cite{schmeidler1971fair}.
A fair allocation is defined as a Pareto Efficient allocation in which everyone prefers their own bundle to other players' bundles, which is exactly the notion of envy-freeness proposed in \cite{foley1967resource}.

The existence of Pareto Efficient and envy-free allocations has been studied in both the divisible and indivisible cases.

When items are divisible, previous work \cite{varian1973equity, svensson1983existence, diamantaras1992equity, svensson1994sigma, vohra1992equity} showed that Pareto Efficient and envy-free allocations exist under a variety of assumptions, including that utility functions are strictly monotone, continuous, or convex. In contrast, Vohra~\cite{vohra1992equity} showed that when the economy has increasing-marginal-returns, there exist cases such that no Pareto Efficient and envy-free allocation exists. Also, Maniquet~\cite{maniquet1999strong} gave an example with two items and three players for which there is no Pareto Efficient and envy-free allocation.

In the indivisible setting, for the case of mixed allocations, Bogomolnaia and Moulin~\cite{bogomolnaia2001new} introduced the Probabilistic Serial mechanism and showed this new mechanism results in an \emph{ordinally efficient} expected matching which is envy-free in their setting. Ordinal efficiency is a notion which is slightly weaker than Pareto Efficiency. Budish et al.~\cite{budish2013designing} gave a Pareto Efficient and envy-free allocation when the allocation constraints satisfy a bihierarchy assumption which applies to multi-item allocation problems with possibly non-linear utility functions.

For the deterministic case, because of the simple counter-example mentioned above, researchers have proposed many other notions of fairness. The two most closely related notions are \emph{EF1} (\emph{envy-free up to one good}) \cite{budish2011combinatorial} and \emph{EFX} (\emph{envy-free up to any good}) \cite{caragiannis2016unreasonable}. Recall that the idea in the definition of envy-freeness is that each player will compare their bundle to those of the other players. These alternate notions also have players compare their bundle to the other players' bundles, but in \emph{EF1}, players delete their favorite item from the other bundle before doing the comparison, and in  \emph{EFX}, players will not envy another bundle after deleting their least favorite item. Lipton et al.~\cite{lipton2004approximately}  showed that an \emph{EF1} allocation always exists. For the \emph{EFX} allocation, Plaut and Roughgarden~\cite{plaut2018almost} showed that in some situations (utility functions are identical or additive) existence is guaranteed, while for general utility functions, there exist examples such that no \emph{EFX} allocation is Pareto Efficient.

In addition, Dickerson et al.~\cite{dickerson2014computational} showed that if the number of items is at least a logarithmic factor larger than the number of players, then with high probability, an envy-free allocation exists.

Recently, Richter and Rubinstein \cite{richter2018normative} introduced the Normative Equilibrium. They considered when deterministic Pareto Efficient and envy-free allocations exist
in this setting.

Other fairness notions include Nash Social Welfare \cite{caragiannis2016unreasonable, cole2018approximating, cole2016convex}, max-min fairness \cite{le2005rate}, and CEEI \cite{varian1973equity}.

There has been considerable recent work \cite{cole2018approximating, anari2018nash, bei2016computing, birnbaum2011distributed, cheung2019tatonnement, devanur2002market, garg2018approximating, lee2017apx} on the computational complexity of computing the Nash  Social Welfare, both exactly and approximately, for divisible and indivisible items.

Communication complexity has a long history in game theory \cite{hart2007communication, babichenko2017communication, goldberg2014communication, conitzer2004communication, goos2018near}, especially in allocation problem \cite{nisan2002communication, grigorieva2006communication, plaut2019communication, branzei2019communication}. Nisan and Segal \cite{nisan2002communication} studied the communication complexity of maximizing social welfare and the supporting prices. Plaut and Roughgarden \cite{plaut2019communication} focused on the communication complexity of finding each of a deterministic envy-free and a proportional division with indivisible goods. Br{\^a}nzei and Nisan \cite{branzei2019communication} looked at the communication complexity of the cake cutting problem.

\section{Notation and Results}

There are $m$ items and $n$ players. Each player has a utility function $u_i(x_i)$ on each subset $x_i \subseteq \{1, 2, \cdots, m\}$. And there are $k$ allocations, $A^{(1)}, A^{(2)}, \cdots, A^{(k)}$. Allocation $A^{(j)}$ allocates $A^{(j)}_i$ to player $i$, where $A^{(j)}_i \subseteq \{1, 2, \cdots, m\}$ and $A^{(j)}_i \cap A^{(j)}_{h} = \emptyset$ for any $h$ and $i$ with $h \neq i$.

We say the allocation set is \emph{swappable} if and only if for any allocation $A^{(j)}$ and any pair of players $g$ and $h$, there exists an allocation $A^{(l)}$ such that $A^{(j)}_i = A^{(l)}_i$ for any $i \neq \{g, h\}$, $A^{(j)}_{g} = A^{(l)}_{h}$, and $A^{(j)}_{h} = A^{(l)}_{g}$.

We define a mixed allocation to be a probability distribution on the possible allocations: $\mathbf{p} = (p_1, p_2, \cdots, p_k) \in P$ such that $\sum_j p_j = 1$ and $p_j \geq 0$ for any $j$. Given $\mathbf{p}$, player $i$'s expected utility is $\sum_j p_j u_i(A^{(j)}_i)$.

A mixed allocation $\mathbf{p}$ is \emph{Pareto Efficient} (PE) if there is no another mixed allocation $p'$ such that for all $i$, $\sum_j p'_j u_i(A^{(j)}_i) \geq \sum_j p_j u_i(A^{(j)}_i)$ and there exist one $i$ such that this inequality is strict.

A mixed allocation $\mathbf{p}$ is \emph{envy-free} (EF) if for every pair  $i$ and $i'$ of players, $\sum_j p_j u_i(A^{(j)}_i) \geq \sum_j p_j u_i(A^{(j)}_{i'})$.

Now, we present our main theorem.

\begin{theorem}\label{thm::exist}
If the allocation set is swappable, then there exists a Pareto Efficient and envy-free mixed allocation.
\end{theorem}

Next, we will look at the communication complexity of finding a Pareto Efficient and envy-free allocation. In the communication complexity part, we assume that allocation set includes all possible partitions of the items.

\paragraph{Communication complexity setting}
In this setting, we assume that each player $i$ only knows her own utility function $u_i(\cdot)$. She has no knowledge of others' utility functions. The problem is to determine the minimum number of bits they need to communicate with each other in order to calculate a Pareto Efficient and envy-free allocation, which means that, after communication and individual calculation,  every player will agree on one allocation which is both Pareto Efficient and envy-free.

In this setting, we assume that for every $S$ and $i$, $u_i(S)$ can be represented with a polynomial number of bits. Given this assumption, one easy conclusion is that there always exists a protocol  which uses an exponential number of bits (exponential in the number of items) and outputs an allocation which is Pareto Efficient and envy-free. This is true because there would be no difficulty for all players to calculate one Pareto Efficient and envy-free allocation if each player knows everyone else's utility function. An exponential number of bits of communication suffices to achieve this.

So, our question is if it is possible there exists a protocol that needs only a polynomial number of bits.

We show the following negative result.
\begin{theorem} \label{thm::cc::1}
  Any protocol that calculates Pareto Efficient and envy-free allocation for two players with submodular uility functions needs an exponential number of bits.
\end{theorem}

The definition of a submodular utility function follows.
\begin{definition}
  $u(\cdot)$ is a submudular function if and only if for any $X \subseteq Y$ and any element $e$,
  \begin{align*}
    u(X \cup \{e\}) - u(X) \geq u(Y \cup \{e\}) - u(Y).
  \end{align*}
\end{definition}

\section{Proof of Theorem~\ref{thm::exist}}
WLOG, we assume that $1 \leq u_i(x_i) \leq 2$ for all $i$ and $x_i$.
We will use a fixed point argument. To this end, we construct a mapping from $P \times W$ to itself. Here, $P$ is the set of mixed allocations and $W$ is the set of weighted vectors $\{\mathbf{w} = (w_1, w_2, \cdots, w_n) | \sum_i w_i = 1~\text{and}~ w_i \geq \epsilon \}$. We will specify $\epsilon$ later. Now we construct a mapping $\Phi$ from $(
\mathbf{p}, \mathbf{w})$ to $(\mathcal{P}(\mathbf{w}), \varpi(\mathbf{p}, \mathbf{w}))$, where $\mathcal{P}(\mathbf{w})$ is a subset of $P$ and $\varpi(\mathbf{p}, \mathbf{w}) \in W$.
\begin{align*}
&\mathcal{P}(\mathbf{w}) = \{\mathbf{p}' | \mathbf{p}' \in P ~\text{and}~ \mathbf{p}' \in \arg \max \sum_i w_i \sum_j p'_j u_i(A^{(j)}_i) \}; \\
&\varpi(\mathbf{p}, \mathbf{w}) = \mathtt{proj}_{W}(\nu(\mathbf{p}, \mathbf{w})); \numberthis\label{def::varpi}\\
&\mbox{where~} \nu_i(\mathbf{p}, \mathbf{w}) =  w_i + \frac{\max_{h} \sum_j p_j u_i(A^{(j)}_{h})}{\sum_{i'} \max_{h} \sum_j p_j u_{i'}(A^{(j)}_{h})} -  \frac{\sum_j p_j u_i(A^{(j)}_{i})}{\sum_{i'}  \sum_j p_j u_{i'}(A^{(j)}_{i'})}. \numberthis \label{def::nu}
\end{align*}
\begin{lemma} \label{lem::fixed-point}
There exists a fixed point, $(\mathbf{p}^*, \mathbf{w}^*)$, such that $\mathbf{p}^* \in \mathcal{P}(\mathbf{w}^*)$ and $\varpi(\mathbf{p}^*, \mathbf{w}^*)  = \mathbf{w}^*$.
\end{lemma}
We prove Lemma~\ref{lem::fixed-point} via the following three claims.
\begin{claim}\label{clm::fixed::1}
Let $\mathcal{A}(\mathbf{w}) = \{j | A^{(j)}~\textit{maximizes}~\sum_i w_i u_i(A^{(j)}_i)~\textit{over all allocations}\}$. Then, $\mathcal{P}(\mathbf{w})$ is a simplex on $A(\mathbf{w})$, such that for any $\mathbf{p}'' \in \mathcal{P}(\mathbf{w})$, $p_j'' > 0$ only if $j \in \mathcal{A}(\mathbf{w})$ (and, of course $\sum_j p_j'' = 1$ and $p_j'' \geq 0$):
\begin{align*}
\mathcal{P}(\mathbf{w}) = \left\{\mathbf{p}'' \Bigg| \left(p''_j > 0 \rightarrow  j \in \mathcal{A}(\mathbf{w})\right)\cap \left(\forall j. ~ p''_j \geq 0 \right) \cap \left( \sum_j p''_j = 1\right) \right\}.
\end{align*}
\end{claim}
\begin{proof}[Proof of Claim~\ref{clm::fixed::1}]
It's not hard to see that in the definition of $\mathcal{P}(\mathbf{w})$, we can rewrite $\sum_i w_i \sum_j p'_j u_i(A^{(j)}_i)$ as $\sum_j p'_j \sum_i w_i u_i(A^{(j)}_i)$. So any probability $p'_j >0$ on an allocation $A^{(j)}$ that does not maximize $\sum_i w_i u_i(A^{(j)}_i)$ will contradict the definition of $\mathcal{P}(\mathbf{w})$.
\end{proof}
\begin{claim}\label{clm::fixed::2}
Give a mixed allocation $\mathbf{p}' \in P$, the set of $\mathbf{w}$ such that $\mathbf{p}' \in \mathcal{P}(\mathbf{w})$ is a convex closed set.
\end{claim}
\begin{proof}[Proof of Claim~\ref{clm::fixed::2}]
Given $\mathbf{p}'$, the condition $\mathbf{p}' \in \mathcal{P}(\mathbf{w})$ means that for any allocation $j$ for which $\mathbf{p}'$ has positive probability, its weighted sum of utilities $\sum_i w_i u_i(A^{(j)}_i)$ is no smaller than that for any other allocation. More formally, for any $j$ and $l$ such that $p_j >0$, $\sum_i w_i u_i(A^{(j)}_i) \geq \sum_i w_i u_i(A^{(l)}_i)$. The remaining conditions on $\mathbf{w}$, $\forall i. ~w_i \geq \epsilon$ and $\sum_i w_i = 1$, are also linear. This means the feasible set of $\mathbf{w}$ is a convex closed set.
\end{proof}
\begin{claim}\label{clm::fixed::3}
For any series $(\mathbf{w}^{(t)}, \mathbf{p}^{(t)})$ with $\lim_{t \rightarrow \infty} \mathbf{w}^{(t)} = \mathbf{\overline{w}}$ and $\lim_{t \rightarrow \infty} \mathbf{p}^{(t)} = \mathbf{\overline{p}}$, if for every $t$, $\mathbf{p}^{(t)} \in \mathcal{P}(\mathbf{w}^{(t)})$, then $\mathbf{\overline{p}} \in \mathcal{P}(\mathbf{\overline{w}})$. In other words, $\mathcal{P}(\overline{\mathbf{w}})$ has a closed graph.
\end{claim}
\begin{proof}[Proof of Claim~\ref{clm::fixed::3}]
Consider the set $\mathcal{S} = \left\{j | \mathbf{\overline{p}}_j > 0 \right\}$. By Claim~\ref{clm::fixed::1}, we only need to show that $A^{(j)}$ maximizes $\sum_i \overline{w}_i u_i(A_i^{(j)})$ for all $j \in \mathcal{S}$. We observe that there exists a $t'$ such that $\mathbf{p}^{(t)}_{j} > 0$ for all $j \in \mathcal{S}$ and all $t \geq t'$, as $\lim_{t \rightarrow \infty} \mathbf{p} = \mathbf{\overline{p}}$ and the dimension of $\mathbf{p}$, which is the number of allocations,  is finite. Consider the set $\mathcal{W}$ of $\mathbf{w}$ such that $A^{(j)}$ maximizes $\sum_i w_i u_i(A_i^{(j)})$  for all $j \in \mathcal{S}$. Then, we only need to show that $\mathbf{\bar{w}} \in \mathcal{W}$. Since $\mathbf{p}^{(t)} \in \mathcal{P}(\mathbf{w}^{(t)})$, $\mathbf{p}^{(t)} \in \arg \max \sum_j p_j^{(t)} \sum_i w_i^{(t)} u_i(A_i^{(j)})$. By our observation, $p_j^{(t)} > 0$ for $t > t'$ and $j \in \mathcal{S}$; thus $A^{(j)}$ maximizes $\sum_i w_i^{(t)}u_i(A^{(j)}_i)$ for all $j \in \mathcal{S}$. This implies $\mathbf{w}^{(t)} \in \mathcal{W} $ for any $t \geq t'$. Furthermore, $\mathcal{W}$ is convex and closed by Claim~\ref{clm::fixed::2}, which implies $\mathbf{\overline{w}} \in \mathcal{W}$.  By Claim~\ref{clm::fixed::1}, $\mathbf{\overline{p}} \in \mathcal{P}(\mathbf{\overline{w}})$.
\end{proof}
\begin{proof}[Proof of Lemma~\ref{lem::fixed-point}]
  For each $(\mathbf{p}, \mathbf{w})$, the mapping $\Phi(\mathbf{p},\mathbf{w}) = (\mathcal{P}(\mathbf{w}), \varpi(\mathbf{p}, \mathbf{w}))$ is a convex set as $\mathcal{P}(\mathbf{w})$ is convex by Claim~\ref{clm::fixed::1} and $\varpi(\mathbf{p}, \mathbf{w})$ is a single point. It is non-empty as $\mathcal{A}(\mathbf{w})$ is non-empty, and by Claim~\ref{clm::fixed::1} the corresponding $\mathcal{P}(\mathbf{w})$ is non-empty. By Claim~\ref{clm::fixed::3}, $\Phi$ has a closed graph. Hence Kakutani's fixed point theorem can be applied.
\end{proof}
We now conclude that for any fixed point $(\mathbf{p}^*, \mathbf{w}^*)$, $\mathbf{p}^*$ is a Pareto Efficient allocation.
\begin{claim}\label{clm::PE}
If $(\mathbf{p}^*, \mathbf{w}^*)$ is a fixed point of $\Phi$, then $\mathbf{p}^*$ is a Pareto Efficient allocation.
\end{claim}
\begin{proof}
Note that the fact that $\mathbf{p}^* \in \mathcal{P}(\mathbf{w^*})$ means $\mathbf{p}^*$ maximizes $\sum_i w^*_i \sum_j p^*_j u_i(A^{(j)}_i)$. So, there cannot be another $\mathbf{p}$ such that for every $i$, $\sum_j p_j u_i(A^{(j)}_i) \geq \sum_j p^*_j u_i(A^{(j)}_i)$, with the inequality being strict for some $i$.
\end{proof}
The following two lemmas prove Theorem~\ref{thm::exist}. Recall that the definition of $\varpi(\cdot)$ and $\nu(\cdot)$ are given in \eqref{def::varpi} and \eqref{def::nu}.
\begin{lemma} \label{lem::IN}
If $\nu(\mathbf{p}^*, \mathbf{w}^*) \in W$, then $\mathbf{p}^*$ is a Pareto Efficient and envy-free allocation.
\end{lemma}
\begin{lemma} \label{lem::not::in}
if $\nu(\mathbf{p}^*, \mathbf{w}^*)$ is not in $W$, then $w^* \neq \varpi(\mathbf{p}^*, \mathbf{w}^*)$.
\end{lemma}
\begin{proof}[Proof of Theorem~\ref{thm::exist}]
By Lemma~\ref{lem::fixed-point}, the mapping $\Phi$ has a fixed point $(\mathbf{p}^*, \mathbf{w}^*)$. By Lemma~\ref{lem::not::in}, $\nu(\mathbf{p}^*, \mathbf{w}^*) \in W$, since otherwise $(\mathbf{p}^*, \mathbf{w}^*)$ would not be a fixed point. The theorem now follows from Lemma~\ref{lem::IN}.
\end{proof}

Next, we start to prove Lemma~\ref{lem::IN}. We first construct an \emph{envy graph}  $(V, E)$ based on $\mathbf{p}^*$. $V$ is the set of players and the directed edge from $i$ to $i'$, $(i, i')$, belongs to $E$ if and only if $i$ envies $i'$, which means $\sum_j p_j^* u_i(A^{(j)}_i) < \sum_j p_j^* u_i(A^{(j)}_{i'})$.
\begin{claim} \label{clm::graph::acyclic}
Let $\mathbf{p}^*$ be a Pareto Efficient mixed allocation. The corresponding envy graph is acyclic.
\end{claim}
\begin{proof}[Proof of Claim~\ref{clm::graph::acyclic}]
If the graph has a cycle, then we can improve everyone's utility functions in this cycle by exchanging the allocations along the cycle, contradicting Pareto Efficiency. This argument is using the swappable property.
\end{proof}
Given the Pareto Efficient mixed allocation $\mathbf{p}^*$, we define the set of envy-free players to be $\mathcal{I}(\mathbf{p}^*) = \{ i | \textit{player i does not envy player h for all h} \}$.
\begin{claim} \label{clm::no::envy::set}
Suppose $\mathbf{p}^*$ be a Pareto Efficient mixed allocation. Then $\mathcal{I}(\mathbf{p}^*)$  is not empty.
\end{claim}
\begin{proof}[Proof of Claim~\ref{clm::no::envy::set}]
This follows from the fact that the envy graph is acyclic.
\end{proof}

\begin{claim}\label{clm::less::ineq}
Let $\mathbf{p}^*$ be a Pareto Efficient mixed allocation. Then,  for any $i$ in $\mathcal{I}(\mathbf{p}^*)$,
$$\frac{\max_{h} \sum_j p_j u_i(A^{(j)}_{h})}{\sum_{i'} \max_{h} \sum_j p_j u_{i'}(A^{(j)}_{h})} \leq  \frac{\sum_j p_j u_i(A^{(j)}_{i})}{\sum_{i'}  \sum_j p_j u_{i'}(A^{(j)}_{i'})},$$
and  equality holds if and only if $\mathbf{p}^*$ is envy-free.
\end{claim}
\cite{svensson1983existence} gives a result similar to Claim~\ref{clm::less::ineq}. Here we provide a simple proof for completeness.
\begin{proof}[Proof of Claim~\ref{clm::less::ineq}]
The inequality follows from the following two facts:
\begin{itemize}
\item $\max_{h} \sum_j p_j u_{i'}(A^{(j)}_{h}) \geq \sum_j  p_j u_{i'}(A^{(j)}_i)$;
\item for any $i \in \mathcal{I}(\mathbf{p}^*)$, $\max_{h} \sum_j p_j u_i(A^{(j)}_{h}) = \sum_j  p_j u_i(A^{(j)}_i)$.
\end{itemize}
Equality holds if and only if for all players, $\max_{h} \sum_j p_j u_{i'}(A^{(j)}_{h}) = \sum_j  p_j u_{i'}(A^{(j)}_i)$ and thus no one envies anyone else.
\end{proof}
\begin{proof}[Proof of Lemma~\ref{lem::IN}]
It is not hard to see that if $\nu(\mathbf{p}^*, \mathbf{w}^*) \in W$ and $\mathbf{w}^*$ is a fixed point, then $w^* = \varpi(\mathbf{p}^*, \mathbf{w}^*) = \nu(\mathbf{p}^*, \mathbf{w}^*)$, and $\frac{\max_{h} \sum_j p_j u_i(A^{(j)}_{h})}{\sum_{i'} \max_{h} \sum_j p_j u_{i'}(A^{(j)}_{h})} =  \frac{\sum_j p_j u_i(A^{(j)}_{i})}{\sum_{i'}  \sum_j p_j u_{i'}(A^{(j)}_{i'})}$. In addition, by Claim~\ref{clm::PE}, $\mathbf{p}^*$ is a Pareto Efficient allocation and so from Claim~\ref{clm::less::ineq}, $\mathbf{p}^*$ is envy-free.
\end{proof}
Next, we will show Lemma~\ref{lem::not::in}, namely that if $\nu(\mathbf{p}^*, \mathbf{w}^*)$ is not in $W$, then $w^* \neq \varpi(\mathbf{p}^*, \mathbf{w}^*)$. In particular, we show that there exist an $i$, such that $\varpi_i(\mathbf{p}^*, \mathbf{w}^*) < w_i^*$. The idea is that the allocation is not envy-free as $\nu(\mathbf{p}^*, \mathbf{w}^*)$ is not in $W$. But there must exist one player  who is in the set $\mathcal{I}(\mathbf{p}^*)$, which is the set of envy-free players, as the envy graph is acyclic. In addition, for this player $i$, $\nu_i(\mathbf{p}^*, \mathbf{w}^*) < w_i^*$ by Claim~\ref{clm::less::ineq}. However, since $\varpi(\mathbf{p}^*, \mathbf{w}^*) = \mathtt{proj}_W(\nu(\mathbf{p}^*, \mathbf{w}^*))$, it is still  possible that $\varpi_i(\mathbf{p}^*, \mathbf{w}^*) = w_i^*$. Actually, we show the following two claims.
\begin{claim} \label{clm::proj}
$\varpi_i(\mathbf{p}^*, \mathbf{w}^*) \leq \max \{\nu_i(\mathbf{p}^*, \mathbf{w}^*), \epsilon \}$ for all $i$.
\end{claim}
\begin{claim}\label{clm::exist::no::envy}
Let $\mathbf{p}^* \in \mathcal{P}(\mathbf{w})$. If $\epsilon < \frac{\rho^n}{n}$, then there exists a player $i$ such that $w_i > \epsilon$ and $i \in \mathcal{I}(\mathbf{p}^*)$.
\end{claim}
\begin{proof}[Proof of Lemma~\ref{lem::not::in}]
Since $\nu(\mathbf{p}^*, \mathbf{w}^*)$ is not in $W$, then  $\max_{\bar{i}} \sum_j p_j u_i(A^{(j)}_{\bar{i}}) \neq \sum_j  p_j u_i(A^{(j)}_i)$ for some $i$. So $\mathbf{p}^*$ is not an envy-free allocation, which implies that for $i \in \mathcal{I}(\mathbf{p}^*)$,\begin{align*}
\frac{\max_{\bar{i}} \sum_j p_j u_i(A^{(j)}_{\bar{i}})}{\sum_{i'} \max_{\bar{i}} \sum_j p_j u_{i'}(A^{(j)}_{\bar{i}})} < \frac{\sum_j p_j u_i(A^{(j)}_{i})}{\sum_{i'}  \sum_j p_j u_{i'}(A^{(j)}_{i'})}.  \numberthis \label{ineq::strict}
\end{align*}  From Claim~\ref{clm::exist::no::envy}, we know that there exists one player $i^* \in \mathcal{I}(\mathbf{p}^*)$ with  $w^*_{i^*} > \epsilon$. Therefore, by \eqref{ineq::strict} and \eqref{def::nu}, $\nu_{i^*}(\mathbf{p}^*, \mathbf{w}^*) < w^*_{i^*}$. By Claim~\ref{clm::proj}, $\varpi_{i^*}(\mathbf{p}^*, \mathbf{w}^*) \leq \max\{\nu_{i^*}(\mathbf{p}^*, \mathbf{w}^*), \epsilon\} < w^*_{i^*}$, the result follows.
\end{proof}
The proof of Claim~\ref{clm::proj} follows readily from the definition of $\mathtt{proj}$.
\begin{proof}[Proof of Claim~\ref{clm::proj}]
For simplicity, let $x^* = \varpi(\mathbf{p}^*, \mathbf{w}^*)$ and $y^* = \nu(\mathbf{p}^*, \mathbf{w}^*)$. So we need to show $x_i^* \leq \max\{ y_i^*, \epsilon\}$ for all $i$. Then, by \eqref{def::varpi}, $x^* = \mathtt{proj}_W y^*$ means $x^*$ is the result of the following optimization program:
$\min_x \frac{1}{2}\|x - y^*\|^2$ such that $x \in W$. Note that $W = \{(w_1, \cdots, w_n | \sum_i w_i = 1 \text{ and } w_i \geq \epsilon \}$. So, $x^*$ is the optimal solution for the following convex program:
\begin{align*}
&\min_x \frac{1}{2} \|x - y^*\|^2\\
&~\text{s.t.} \sum_i x_i = 1 \\
& ~~~~~~ x_i \geq \epsilon~~~~~~\forall i
\end{align*}
The Lagrange form is $\frac{1}{2} \|x - y^*\|^2 - \lambda (\sum_i x_i - 1) - \sum_i \beta_i(x_i - \epsilon)$. From the KKT conditions, we know that $x^*_i - y^*_i - \lambda - \beta_i = 0$, $\beta_i \geq 0$, $x^*_i \geq \epsilon$, and $\beta_i (x^*_i - \epsilon) = 0$; also $\sum_i x^*_i = 1$. From the construction of $\nu(\mathbf{p}^*, \mathbf{w}^*)$, we know that $\sum_i y^*_i = \sum_i w_i^* = 1$ by \eqref{def::nu}. By summing the KKT condition $x^*_i - y^*_i - \lambda - \beta_i = 0$ over all $i$, we obtain $\lambda= -  \frac{1}{n} \sum_i \beta_i \leq 0$. If $x_i^* = \epsilon$ then the result holds. Otherwise, $\beta_i =0$ as $\beta_i (x_i^* - \epsilon) = 0$; this implies $x_i^*= y_i^* + \lambda \leq y_i^*$, as $\lambda \leq 0$, and again the result holds.
\end{proof}
The proof of Claim~\ref{clm::exist::no::envy} needs the following additional claim.
\begin{claim}\label{clm::gap}
Let $\mathbf{p} \in \mathcal{P}(\mathbf{w})$ be a Pareto Efficient allocation. Suppose that  $w_h \leq \rho w_i$,  where $$\rho = \frac{1}{2} \min_{\substack{i, h, j \\ u_i(A^{(j)}_i) < u_i(A^{(j)}_h) \text{ and  } \\ u_h(A^{(j)}_i) < u_h(A^{(j)}_h)}}  \frac{u_i(A^{(j)}_h) - u_i(A^{(j)}_i)}{u_h(A^{(j)}_h) - u_h(A^{(j)}_i)}.$$
Then player $i$ will not envy player $j$
\end{claim}
The intuition for this claim is that if we put a relatively higher weight on player $i$ rather than on player $h$ and maximize the sum of weighted utilities, then player $i$ should not envy player $h$. This proof uses the swappable property. One thing to note is that $\rho$ is well defined and positive as the total number of allocations is finite.
\begin{proof}
Consider $\mathcal{A(\mathbf{w})}$ as defined in Claim~\ref{clm::fixed::1}.  Suppose that for any pure allocation in this set, player $i$ does not envy player $h$. Then player $i$ will not envy player $h$ in the mixed allocation $\mathbf{p} \in \mathcal{P}(\mathbf{w})$. We show by contradiction that for any pure allocation in this set, player $i$ will not envy player $h$. Suppose player $i$ envies player $h$ in an allocation $A^{(j)}$. Then, since $j \in \mathcal{A(\mathbf{w})}$ and the allocation set is swappable,
\begin{align*}
w_i u_i(A^{(j)}_i) + w_h u_h(A^{(j)}_h) \geq w_i u_i(A^{(j)}_h) + w_h u_h(A^{(j)}_i).
\end{align*}
Thus
\begin{align*}
w_h \left[ u_h(A^{(j)}_h) - u_h(A^{(j)}_i) \right] \geq w_i \left[u_i(A^{(j)}_h) - u_i(A^{(j)}_i) \right]. \numberthis \label{ineq::fix::1}
\end{align*}
Since player $i$ envies player $h$,
\begin{align*}
u_i(A^{(j)}_i) < u_i(A^{(j)}_h). \numberthis \label{ineq::fix::2}
\end{align*}
From \eqref{ineq::fix::1}, \eqref{ineq::fix::2} and the fact that $\mathbf{w}$ is strictly positive,
\begin{align*}
u_h(A^{(j)}_i) < u_h(A^{(j)}_h).
\end{align*}
Since $0 < w_h \leq \rho w_i$ and from \eqref{ineq::fix::1},
\begin{align*}
\rho (u_h(A^{(j)}_h) - u_h(A^{(j)}_i)) \geq u_i(A^{(j)}_h) - u_i(A^{(j)}_i),
\end{align*}
which contradicts the definition of $\rho$.
\end{proof}
\hide{\begin{claim}\label{clm::exist::no::envy}
Given $\mathbf{p}^* \in \mathcal{P}(\mathbf{w})$, if $\epsilon < \frac{\rho^n}{n}$, then there exists a player $i$ such that $w_i > \epsilon$ and $i \in \mathcal{I}(\mathbf{p}^*)$.
\end{claim}}
Given Claim~\ref{clm::gap}, Claim~\ref{clm::exist::no::envy} can be shown by selecting a small enough $\epsilon$.
\begin{proof}[Proof of Claim~\ref{clm::exist::no::envy}]
Since $\epsilon < \frac{\rho^n}{n}$, there exists an $\xi$ such that $\frac{\rho}{n} > \xi > \epsilon$ and, for each $i$, $w_i$ is not in the interval between $\xi$ and $\frac{\xi}{\rho}$. Therefore, we can separate all the players into two sets $D = \{ i | w_i < \xi\}$ and $U = \{i | w_i > \frac{\xi}{\rho}\}$. So for every player $i$ in $U$, $w_i > \epsilon$. Note that $D \cap U = \emptyset$, $D \cup U = \{1, 2, \cdots, n\}$, and $U$ is not empty since there exists one player $i$ with $w_i \geq \frac{1}{n}$. By Claim~\ref{clm::gap}, we know that players in $U$ will not envy  players in $D$, and by Claim~\ref{clm::graph::acyclic} we know that the envy graph is acyclic, thus there must be a player in $U$ that is envy-free, i.e. in $\mathcal{I}(\mathbf{p}^*)$. The result follows.
\end{proof}

\section{Proof of Theorem~\ref{thm::cc::1}}
\hide{In this section, we will look at the communication complexity of finding a Pareto Efficient and envy-free allocation.

\subsection{Setting}

In this setting, we assume that each player $i$ only knows her own utility function $u_i(\cdot)$. She has no knowledge of others' utility functions. The problem is to find out the minimum number of bits they need to communicate with each other in order to calculate a Pareto Efficient and envy-free allocation, which means that, finally, after communicating with each others and doing some calculations by themselves,  every player will agree on one allocation which is Pareto Efficient and envy-free.

In this setting, we assume that for every $S$ and $i$, $u_i(S)$ can be represented with a polynomial number of bits. Given this assumption, one easy conclusion is that there always exists a protocol  which uses an exponential number of bits (exponential in the number of items) and outputs an allocation which is Pareto Efficient and envy-free. This is true because there would be no difficulty for all players to calculate one Pareto Efficient and envy-free allocation if each player knows everyone else's utility function. An exponential number of bits of communication suffices to achieve this.

So, our question is if it is possible there exists a protocol that needs only a polynomial number of bits.

We show the following negative result.
\begin{theorem} \label{thm::cc::1}
  Any protocol that calculates Pareto Efficient and envy-free allocation for two players with submodular uility functions needs an exponential number of bits.
\end{theorem}

The definition of a submodular utility function follows.
\begin{definition}
  $u(\cdot)$ is a submudular function if and only if for any $X \subseteq Y$ and any element $e$,
  \begin{align*}
    u(X \cup \{e\}) - u(X) \geq u(Y \cup \{e\}) - u(Y).
  \end{align*}
\end{definition}}

In order to prove this result, we make a reduction to the Disjointness problem.
\paragraph{Disjointness problem} Suppose there are two players. Each player has a string of bits of length $\mathcal{L}$. Suppose $x_i$ is player $i$'s string. The players want to determine if there exists a $j$ such that $x_{1j} = x_{2j} = 1$. The complexity question asks how many bits they need to communicate in order to determine the result.

\begin{theorem}[\cite{kalyanasundaram1992probabilistic, razborov1990distributional}]
  Any protocol for the Disjointness problem needs to communicate $\Omega(\mathcal{L})$ bits.
\end{theorem}

\begin{proof}[Proof of Lemma~\ref{thm::cc::1}]
We follow a similar approach to \cite{plaut2019communication}. We give a reduction to the Disjointness problem.

Suppose there are $m = 2p$ items and two players. Each player receives a string of bits.
Let $x_1$ and $x_2$, respectively, be the strings player $1$ and player $2$ receive and suppose they satisfy $|x_1| = |x_2| = \frac{1}{2} \chs{2p}{p}$.

Let $\mathcal{T} = \{ T_{1}, T_{2}, \cdots, T_{r}\}$ be a set of allocations, where $r = \frac{1}{2} \chs{2p}{p}$ and $T_{j} = (T_{j1}, T_{j2})$. Informally speaking, $T_{j1}$ is a subset of items which includes item $1$ and $T_{j2}$ is the complementary set of $T_{j1}$. It's easy to verify that there are in total $\frac{1}{2} \chs{2p}{p}$ possible choices of $T_{j1}$; this is the reason that $r = \frac{1}{2} \chs{2p}{p}$. Formally, $T_{j} = (T_{j1}, T_{j2})$ is an allocation such that $T_{j1} \cup T_{j2} = \{1, 2, \cdots, 2p\}$, $T_{j1} \cap T_{j2} = \emptyset$, $1 \in T_{i1}$ and for any $j \neq l$, $T_{j1} \neq T_{l1}$.

Then, given $x_i$, player $i$'s utility function is defined as follows:
\begin{align*}
  u_i(S) =   \begin{cases}
      3|S|  ~~~ &\mbox{if $|S| < p$} \\
      3 k   ~~~ &\mbox{if $|S| > p$} \\
      3 k   ~~~ &\mbox{if $S = T_{ji}$ for some $j$ and $x_{ij} = 1$} \\
      3 k - 1 ~~~&\mbox{otherwise}.
    \end{cases}
\end{align*}
In the following paragraphs, we prove the following lemma.
\begin{lemma} \label{lem::cc::1}
  If $(x_1, x_2)$ is a no-instance of Disjointness problem, then for all Pareto Efficient and envy-free allocations, the sum of utilities, $u_1 + u_2$, will equal $6p$. Otherwise, if $(x_1, x_2)$ is a yes-instance, then for all Pareto Efficient and envy-free allocations, the sum of utilities will be no larger than $6p - 1$.
\end{lemma}
This lemma actually shows that the communication cost of calculating the utility of the Pareto Efficient and envy-free allocation is at least the communication cost of Disjointness problem with $\mathcal{L} = \frac{1}{2} \chs{2p}{p}$, which is exponential in $m$, the number of items, as $m = 2p$. Note that the communication cost of calculating the utility of an allocation is no larger than the communication cost of calculating the allocation, as, without additional communication cost, one can easily calculate the utility from the allocation. Therefore, this statement proves that the communication cost of calculating the Pareto Efficient and envy-free allocation is exponential on $m$.

\begin{proof}[Proof of Lemma~\ref{lem::cc::1}]
If $(x_1, x_2)$ is a no-instance of the Disjointness problem, then there exists a $j$, such that $x_{1j} = x_{2j} = 1$. This implies there exists a Pareto Efficient and envy-free allocation: player $1$ gets $T_{j1}$ and player $2$ gets $T_{j2}$. In this case, both players will have utility $3p$. Since the most utility any player can get is $3p$,  any allocation which is Pareto Efficient must give both players $3p$ utility. This implies that for any allocation which is Pareto Efficient and envy-free, the sum of the utilities, $u_1 + u_2$, will be $6p$.

If $(x_1, x_2)$ is a yes-instance of the Disjointness problem, then for any $j$, $x_{1j} = 0$ or $x_{2j} = 0$. A simple observation is that any allocation in which the items are not fully allocated is not Pareto Efficient. Then we look at any deterministic fully allocated allocation $(A_1, A_2)$. For any $j$, if $A_j = T_{j1}$ and $x_{1j} = 0$ then $u_1(A_1) = 3p - 1$ and $u_2(A_2) \leq 3p$. So $u_1(A_1) + u_2(A_2) \leq 6p - 1$. The same is true if $A_2 = T_{j2}$ and $x_{2j} = 0$. Otherwise, either $|A_1| < p$ or $|A_2| < p$ and then $u_1(A_1) + u_2(A_2) \leq 6p - 3$.

\end{proof}

The final thing to do is to check that this utility function is a submodular utility function. It's easy to see that our utility function satisfies the submodular condition that for any $X \subseteq Y$ and any element $e$,
$u(X \cup \{e\}) - u(X) \geq u(Y \cup \{e\}) - u(Y)$.

The result follows.
\end{proof}

\section{Discussion}
In this paper, we showed that for any utility functions, a Pareto Efficient and envy-free allocation always exists if the allocation set is swappable. It needs to be noted that for the divisible case, the problem still remains open. Our proof cannot be simply generalized to the divisible case. This is because our $\epsilon$ will tend to $0$ as $\rho$ tends to $0$, in which case $\mathcal{P}(\mathbf{w})$ will no longer ensure Pareto Efficiency as for some $i$, $w_i = 0$.

Another open problem is in communication complexity. One simple fact is that if all players have linear utility functions, then there exists a protocol which requires a polynomial number of bits. However, the communication complexity is unknown for larger classes of utilities, such as gross substitutes.
\bibliographystyle{plain}
\bibliography{sample}

\end{document}